\documentclass[11pt]{article}

\usepackage{amsmath, amssymb, amsfonts, amsthm}
\usepackage{url}
\usepackage{graphicx}
\usepackage{enumerate}
\usepackage{cases}
\usepackage[mathscr]{eucal}
\usepackage{xcolor}

\newcommand \R { \mathbb{R} } 
 
\newcommand \N { \mathbb{N} }
 
\newcommand \C { \mathbb{C} }
\newcommand \Sph { \mathbb{S} }

\DeclareMathOperator*{\argmin}{{arg\, min}}
\DeclareMathOperator{\tr}{tr}

\DeclareMathOperator{\imag}{\mathrm{i}}

\DeclareMathOperator{\spn}{{span}}

\newtheorem{theorem}{Theorem}[section]

\newtheorem{lemma}[theorem]{Lemma}
\newtheorem{proposition}[theorem]{Proposition}

\newtheorem{definition}[theorem]{Definition}
\newtheorem{example}[theorem]{Example}

\newtheorem{remark}[theorem]{Remark}

\begin{document} 

\begin{titlepage}

\begin{center}

\textbf{\large Quantum Measurement Trees, II: \\
Quantum Observables as Ortho-Measurable Functions \\
and Density Matrices as Ortho-Probability Measures}

\medskip
Peter J.\ Hammond: \url{p.j.hammond@warwick.ac.uk} \\
Dept.\ of Economics, University of Warwick, Coventry CV4 7AL, UK.

\medskip
This version: 2025 May 16th; typeset from \url{QMeasTreesBCretaWPr.tex}

\bigskip
\textbf{Abstract}:
\end{center}\noindent
Given a quantum state in the finite-dimensional Hilbert space $ \C^n $,  
the range of possible values of a quantum observable 
is usually identified with the discrete spectrum of eigenvalues 
of a corresponding Hermitian matrix.
Here any such observable is identified with:
(i) an ``ortho-measurable'' function 
defined on the Boolean ``ortho-algebra'' generated by the eigenspaces 
that form an orthogonal decomposition of $ \C^n $;
(ii) a ``numerically identified'' orthogonal decomposition of $ \C^n $.
The latter means that each subspace of the orthogonal decomposition 
can be uniquely identified by its own attached real number, 
just as each eigenspace of a Hermitian matrix 
can be uniquely identified by the corresponding eigenvalue.
Furthermore, any density matrix on $ \C^n $ is identified 
with a Bayesian prior ``ortho-probability'' measure defined on the linear subspaces 
that make up the Boolean ortho-algebra induced by its eigenspaces.
Then any pure quantum state is identified with a degenerate density matrix,
and any mixed state with a probability measure on a set of orthogonal pure states.
Finally, given any quantum observable, 
the relevant Bayesian posterior probabilities of measured outcomes 
can be found by the usual trace formula that extends Born's rule.
\texttt{[193 words]}

\bigskip \noindent


\end{titlepage}

\section{Introduction and Outline}
\subsection{Quantum Measurement Trees}

Following Hammond (2025), this is the second paper related to a research project 
on quantum measurement trees.
These are similar to the decision trees considered by Raiffa (1968), 
which are one-player versions of the games in extensive form 
considered by von Neumann (1928).
Following Hammond (1988, 2022), apart from decision nodes. 
as well as terminal nodes to which consequences are assigned, decision trees can include: 
(i) \emph{chance nodes}, where what Anscombe and Aumann (1963) call a roulette lottery, 
with risk described by ``objective'' or hypothetical probabilities, is resolved;
(ii) \emph{event nodes}, where an uncertain event
with ``subjective'' or personal probabilities occurs, 
depending on the outcome of what Anscombe and Aumann (1963) 
call a ``horse lottery'' of the kind considered by Savage (1954).

In fact, a quantum measurement tree is like a decision tree in which, 
in addition to chance and event nodes:
(i) any decision node has become a preparation node 
at which relevant details of a quantum experiment are determined;
(ii) at any terminal node, a possible real measurement of some quantum observable
is determined.
This project addresses the question of how far quantum measurement trees, 
with classical probabilities applying at each chance or event node, 
are able to describe processes that occur in an appropriate Hilbert space 
of solutions to Schr\"odinger's wave equation.%
\footnote{As is well known, these solutions in the complex Hilbert space $ \C^n $ 
typically involve unitary matrices. 
The relevant mathematics of real eigenvalues and orthogonal eigenvectors 
for Hermitian matrices which represent observables and density matrices in $ \C^n $  
turns out to be very similar to that for symmetric matrices which, 
at least in principle if not in physical reality,
could represent observables and density matrices 
whose domain is limited to the real Euclidean space $ \R^n $.}

The first paper presented two examples that can be described without using any Hilbert space.
Here we consider a simple form of quantum measurement tree, 
where there is only one measurement node on each branch of the tree.
The tree starts with a preparation node where an experimental configuration is selected.
This configuration has usually been modelled as the combination of:
(i) a quantum observable in the form of a Hermitian or self-adjoint matrix;
(ii) a quantum state in the form of a density matrix. 
In combination with the quantum state, that Hermitian matrix 
determines at the immediately succeeding measurement chance node 
the specified hypothetical or objective probabilities in a roulette lottery
that determines which eigenspace in its spectral decomposition occurs randomly.

\subsection{Characterizing Observables and Probability Densities}

The contribution of this paper is to describe, 
for the $n$-dimensional complex Hilbert space $ \C^n $, some relevant natural bijections: 
\begin{enumerate}
	\item first, between each pair of the three spaces of:
	\begin{itemize}
		\item Hermitian or self-adjoint $n \times n$ matrices $ \mathbf A$,
		each with a spectrum $ s^{ \mathbf A} \subset \R$ of eigenvalues;
		\item numerically identified orthogonal decompositions of $ \C^n $, 
		as described below in Definition \ref{def:numLabelledOrthDecomp} 
		of Section \ref{ss:orthMeasFns};%
\footnote{The term ``numerically identified'' 
seems preferable to ``numbered'' or ``enumerated'',
since those suggest counting using members of the set $ \N $ of natural numbers 
rather than being identified by a real number.} 
		\item suitably defined functions 
		 \( \C^n \owns \mathbf x \mapsto f^\mathbf A( \mathbf x) \in \R \cup \{*\} \)
		whose co-domain is a suitable one-point extension of the real line,
		and which are ``ortho-measurable'' in the sense of being measurable 
		with respect to the Boolean ``ortho-algebra'' generated by the linear subspaces 
		that form the spectral orthogonal decomposition 
 \( \bigoplus _{ \lambda \in s^{ \mathbf A}} L_\lambda \)
		of $ \C^n $, with the property that 
 \( ( f^\mathbf A) ^{-1} ( \lambda ) = L_\lambda \setminus \{ \mathbf 0 \} \) 
 		for all $ \lambda \in s^{ \mathbf A} $;
	\end{itemize}
 	\item second, between appropriate subspaces of the three spaces specified above
	that correspond to:
	\begin{itemize}
		\item $n \times n$ density matrices $ \boldsymbol \rho $;
		\item orthogonal decompositions of $ \C^n $ 
		that are numerically identified by probability, 
		as described below in Definition \ref{def:probLabelledOrthDecomp}
		of Section \ref{ss:densityMatrices};
		\item suitably defined ``ortho-probability'' measures which are defined 
		on one of the ``ortho-measurable'' spaces defined above.
	\end{itemize}
\end{enumerate}

\subsection{The Quantum Challenge}

As discussed in Hammond (2025), it has been widely recognized 
that the essence of the ``quantum challenge'' can be seen 
as the impossibility of accommodating more than a small set of observable quantum phenomena
within a single contextual probability space in the sense of Kolmogorov.
The main significance of the results concerning bijections in this paper
is that the relevant concepts of measurable set and probability measure 
can remain entirely classical.
To make this possible, however, requires using an appropriate version
of Vorob$'$ev's (1962) notion of a family of probability measures 
defined on a ``multi-measurable space''.
Each of the Boolean or $ \sigma $-algebras involved can be regarded as its own quantum context.
Then staying within just one context is apparently an instance 
of what Griffiths (2002) calls the ``single framework rule''.%
\footnote{See also Hohenberg (2010) as well as Friedberg and Hohenberg (2014).}
Which context or framework is ultimately relevant depends, of course, 
on what quantum observable (or set of quantum observables) is selected
when an experimental configuration is being created. 

\subsection{Outline of Paper}

Section \ref{s:eigenpairs} recapitulates relevant results 
concerning Hermitian or self-adjoint matrices.
In particular, it introduces the notion of an ``eigenpair'' 
combining an eigenvalue with its associated eigenspace of eigenvectors.
It also emphasizes the orthogonal decomposition of the space $ \C^n $ 
into eigenspaces associated with different eigenvalues.

The following Section \ref{s:orthProj} considers orthogonal projection matrices,
along with ortho-partitions, and ortho-algebras, 
especially those induced by the spectral decomposition of a Hermitian matrix.
This leads to the concept of an ortho-measurable function.

After these essential preliminaries, Section~\ref{s:firstBijection} is devoted 
to the first bijections between the three sets of: (i) Hermitian matrices; 
(ii) numerically identified orthogonal decompositions; 
and (iii) ortho-measurable functions. 

Section \ref{s:waveVectors} moves on to probabilities, 
starting with how they relate to wave functions 
--- or rather wave vectors, when time is ignored.
Specifically, a normalized wave vector 
can be regarded as a ``pre-probability'' latent variable in $ \C^n $
that subsequently determines actual probabilities.
The latter are given by Born's rule or, in the special case of an ortho-algebra 
based on the canonical orthonormal basis of $ \C^n $, by the squared modulus rule.

Density matrices are the subject of Section \ref{s:densityMatrices}, 
which also constructs the second family of bijections between: (i) density matrices;
(ii) orthogonal decompositions identified numerically by probability;
and (iii) ortho-probability measures restricted to a measurable space 
which constitutes an ortho-algebra.
Then any pure quantum state is identified with a degenerate density matrix,
and any mixed state that can be represented by a general density matrix 
is identified with a probability mixture of orthogonal pure states.%
\footnote{As is further discussed in Section \ref{ss:pureStates},
this departs from the standard physicists' concept of a mixed state,
which does not require orthogonality of the pure states that have positive probability.} 

The concluding Section \ref{s:conclude} begins by summarizing the key features 
of the quantum measurement tree that this paper has shown how to construct.
These include details of how to construct the measurable and probability metaspaces
of the kind that were used to describe the randomness that occurs  
in the two examples presented in Hammond (2025).
Finally, it also briefly touches on some key issues that deserve answers in future work.

\section{Eigenpairs of Hermitian Matrices in $ \C^n $} \label{s:eigenpairs}

Most of the definitions and results in this section and the next are standard.

\subsection{Hilbert Space and Adjoint Matrices}

The $n$-dimensional linear space $ \C^n $ over the algebraic field $ \C $
has as its typical member the \emph{column $n$-vector} $ \mathbf x = ( x_i )_{i = 1}^n $
whose $n$ components are complex numbers $ x_i \in \C $.
Given any complex number $c = a + \imag b \in \C$, where $a, b \in \R$ and $ \imag ^2 = -1$,
its \emph{complex conjugate} is $ \bar c = a - \imag b \in \C$. 

\begin{definition}
	The \emph{adjoint} $ \mathbf A^* $ 
	of any $m \times n$ matrix $ \mathbf A = ( a_{ij} )_{m \times n} $ is the $n \times m$ 
	transposed conjugate matrix $ \mathbf A^* = ( a^*_{ij} )_{n \times m} $
	whose $ij$ element satisfies $ a^*_{ij} = \bar a_{ji} $, 
	implying that $ \mathbf A^* $ is the transpose $ \bar { \mathbf A}^\top $
	of the matrix $ \bar { \mathbf A}$ whose $ij$ element 
	is the complex conjugate $ \bar a_{ij} $ of the element $ a_{ij} $ of matrix $ \mathbf A$.
\end{definition} 

An important property of adjoint matrices is that, 
for all pairs $ \mathbf A, \mathbf B$ of $n \times n$ matrices,
one has $( \mathbf A \mathbf B)^* = \mathbf B^* \mathbf A^* $.

\smallskip
The space $ \C^n $ becomes a \emph{Hilbert space} when equipped
with the complex-valued \emph{inner product} which is defined
for all pairs $ \mathbf x = ( x_i )_{i = 1}^n $
and $ \mathbf y = ( y_i )_{i = 1}^n $ of column $n$-vectors
by $< \mathbf x, \mathbf y > \ := \sum_{i = 1}^n \bar x_i y_i $. 
This Hilbert space has a real-valued \emph{norm} $ \| \mathbf x \| \ge 0$
whose square is defined for all $n$-vectors $ \mathbf x = ( x_i )_{i = 1}^n $ by
\[ \| \mathbf x \|^2 := \ < \mathbf x, \mathbf x > \ = \sum\nolimits _{i = 1}^n \bar x_i x_i 
	= \sum\nolimits _{i = 1}^n | x_i |^2 \]
\begin{remark}
	Our notation allows 
	the inner product $< \mathbf x, \mathbf y > \ = \sum_{i = 1}^n \bar x_i y_i $
	of any pair of column $n$-vectors $ \mathbf x = ( x_i )_{i = 1}^n $ 
	and $ \mathbf y = ( y_i )_{i = 1}^n $ to be rewritten more concisely
	as the $1 \times 1$ ``matrix'' product $ \mathbf x^* \mathbf y$
	of the $1 \times n$ adjoint row matrix 
	\( \mathbf x^* = \left( ( \bar x_i )_{i = 1}^n \right)^\top = \mathbf { \bar x}^\top \)
	with the $n \times 1$ column matrix $ \mathbf y = ( y_i )_{i = 1}^n $.
\end{remark}

The two vectors $ \mathbf x, \mathbf y \in \C^n $ are \emph{orthogonal} 
just in case $ \mathbf x^* \mathbf y = 0$.

The \emph{unit sphere} of $ \C^n $ is the set
\[ \Sph := \{ \mathbf x \in \C^n \mid \| \mathbf x \|^2 = 1 \} 
		  = \{ \mathbf x \in \C^n \mid \mathbf x^* \mathbf x = 1 \} \]
 
\subsection{Hermitian Matrices and Their Eigenpairs: A Review}

The following two definitions extend to $ \C^n $ the respective definitions 
of symmetric and orthogonal matrices in $ \R^n $.
\begin{definition}
\begin{itemize}
	\item The $n \times n$ matrix $ \mathbf A$ is \emph{Hermitian} or \emph{self-adjoint}
	just in case its adjoint satisfies $ \mathbf A^* = \mathbf A$.
	\item The $n \times n$ matrix $ \mathbf U$ is \emph{unitary} 
	just in case its adjoint satisfies $ \mathbf U^* = \mathbf U^{-1} $.		
\end{itemize}
\end{definition}

The next two definitions are of eigenvalues, eigenvectors, eigenpairs, 
and then the spectrum of a Hermitian matrix.

\begin{definition}
\begin{itemize}
	\item The pair $( \lambda, \mathbf x) \in \C \times( \C^n \setminus \{ \mathbf 0\} )$ 
	is an \emph{eigenpair} of $ \mathbf A$ just in case $ \mathbf {Ax} = \lambda \mathbf x$,
	so $ \lambda $ is an \emph{eigenvalue} 
	and $ \mathbf x \ne \mathbf 0$ is a corresponding \emph{eigenvector}.  	
	\item The \emph{spectrum} of a matrix $ \mathbf A$ is the finite set $ s^{ \mathbf A} $ 
	of its eigenvalues.
\end{itemize}
\end{definition}

The following two results are well known properties 
of the eigenvalues and eigenvectors of a Hermitian matrix.

\begin{proposition}
	Suppose that $ \mathbf A$ is any Hermitian matrix on $ \C^n $
	and that $( \lambda, \mathbf x) $ with $ \mathbf x \ne \mathbf 0$ is any eigenpair.
	Then $ \lambda \in \R$.
\end{proposition}

\begin{proof}
	Because $ \mathbf {Ax} = \lambda \mathbf x$ 
	and so $ \mathbf x^* \mathbf A^* = \bar \lambda \mathbf x^* $, 
	it follows from $ \mathbf A = \mathbf A^* $ that
 \[ ( \lambda - \bar \lambda ) \mathbf x^* \mathbf x
 = \mathbf x^* ( \lambda \mathbf x) - ( \bar \lambda \mathbf x^* ) \mathbf x
 = \mathbf x^* ( \mathbf {Ax} ) - ( \mathbf x^* \mathbf A^* ) \mathbf x 
 = \mathbf x^* ( \mathbf A - \mathbf A^* ) \mathbf x = 0 \]
	But $ \mathbf x \ne \mathbf 0$ implies that $ \mathbf x^* \mathbf x > 0$.
	It follows that $ \lambda = \bar \lambda $, so the eigenvalue $ \lambda $ is real.
\end{proof}

\begin{proposition} \label{orthogSpaces}
	If $( \lambda, \mathbf x) $ are $( \mu, \mathbf y) $ are any two eigenpairs 
	with $ \lambda, \mu \in \R$ and $ \lambda \ne \mu $, then $ \mathbf x^* \mathbf y = 0$. 
\end{proposition}

\begin{proof}
	Because $ \mathbf {Ax} = \lambda \mathbf x$, $ \mathbf {Ay} = \mu \mathbf y$,
	and $ \mathbf A = \mathbf A^* $, one has
\[ ( \lambda - \mu ) \mathbf x^* \mathbf y
	= ( \lambda \mathbf x^* ) \mathbf y - \mathbf x^* ( \mu \mathbf y)
	= ( \mathbf {Ax} )^* \mathbf y - \mathbf x^* ( \mathbf {Ay} ) 
	= \mathbf x^* \mathbf {Ay} - \mathbf x^* \mathbf {Ay} = 0 \]
	But if $ \lambda \ne \mu $ then $ \lambda - \mu \ne 0$, so $ \mathbf x^* \mathbf y = 0$.
\end{proof}

\subsection{Orthogonal Decompositions in \protect{$ \C^n $}}

\begin{definition}
\begin{itemize}
	\item A \emph{linear subspace} $L \subset \C^n $ is a subset
	that is algebraically closed under linear combinations
	--- i.e., if $ \mathbf x, \mathbf y \in L$ and $ \alpha, \beta \in \C$,
	then $ \alpha \mathbf x + \beta \mathbf y \in L$.
	\item Given any linear subspace $L \subset \C^n $,
	define $ \hat L := L \setminus \{ \mathbf 0 \}$ as the set of non-zero vectors in $L$.
	\item Given any set $S \subset \C^n $ of vectors, 
	the set $ \spn S$ is the smallest linear subspace $L \subset \C^n $	
	such that $S \subseteq L$.
	\item Two linear subspaces $L$ and $L'$ of $ \C^n $ are \emph{orthogonal}
	just in case one has $ \mathbf x^* \mathbf y = 0$ 
	for all $ \mathbf x \in L$ and $ \mathbf y \in L'$.
	\item Given any $ \mathbf e \in \C^n \setminus \{ \mathbf 0\} $, let
 \[ [ \mathbf e] := \spn ( \{ \mathbf e \} )
 := \{ \mathbf x \in \C^n \mid \exists c \in \C : \mathbf x = c \, \mathbf e \} \]	
 	denote the one-dimensional linear subspace of $ \C^n $
	that is spanned by the non-zero vector $ \mathbf e$.
\end{itemize}
\end{definition}

\begin{definition}
	The labelled finite family $ \mathcal L^D = \{ L_d \} _{d \in D} $
	of linear subspaces $ L_d $ is:
	\begin{itemize}
		\item \emph{mutually orthogonal} just in case, whenever $d, d' \in D$ with $d \ne d'$,
		 the two spaces $ L_d $ and $ L_{d'} $ are orthogonal;
		\item an \emph{orthogonal decomposition} of $ \C^n $
	just in case the spaces are mutually orthogonal
	and the \emph{direct} or \emph{vector} sum 
 \[ \bigoplus\nolimits _{d \in D} L_d 
    := \left \{ \mathbf x^+ \in \C^n \ \middle \vert \ \forall d \in D; 
	\exists \mathbf x_d \in L_d : 
	\mathbf x^+ = \sum\nolimits _{d \in D} \mathbf x_d \right \} \]
	of all the subspaces $ L_d $ in $ \mathcal L^D $ 
	--- or equivalently, the linear space $ \spn \mathcal L^D $ 
	spanned by $ \{ L_d \} _{d \in D} $	--- is equal to the whole of $ \C^n $.
	\end{itemize}
\end{definition}

\begin{definition}
	For each natural number $m \in \N$, 
	let $ \N_m \subset \N $ denote the set $ \{ 1, 2, \ldots m \}$.
\end{definition}

\begin{example}
	Let $ \{ \mathbf b^k \}_{k \in \N_n } $ be any orthonormal basis of $ \C^n $,
	and let $ \{ M_r \}_{r \in \N_m }$ be any partition of $ \N_n $
	into $m$ pairwise disjoint non-empty sets.
	For each $r \in \N_m $, define $ L_r $ as the linear space of dimension $ \# M_r $
	spanned by the set $ \{ \mathbf b^k \mid k \in M_r \}$ of basis vectors.
	Then $ \bigoplus _{r \in \N_m } L_r $ is an orthogonal decomposition into $m$ subspaces.

	Except in the special case when $ \# M_r = 1$ for all $r \in \N_m $,
	this is entirely different
	from the orthogonal decomposition $ \bigoplus _{k \in \N_n } [ \mathbf b^k ]$
	of $ \C^n $ into the collection of $n$ one-dimensional subspaces $[ \mathbf b^k ]$
	that are each spanned by the respective basis vector $ \mathbf b^k $.
\end{example}

\section{Orthogonal Projections in \protect{$ \C^n $}} \label{s:orthProj}
\subsection{Orthogonal Projection Matrices}

\begin{definition}
	The \emph{orthogonal projection} $ \mathbf x^\perp _L $ of any $ \mathbf x \in \C^n $ 
	onto any linear subspace $L$ of $ \C^n $
	is the unique \emph{closest point} of $L$ to $ \mathbf x$
	--- i.e., it satisfies \( \{ \mathbf x^\perp _L \} 
 := \argmin _{ \mathbf y \in L} \mathbf {(x - y)}^* \mathbf {(x - y)} \).  
\end{definition}

The following is a standard result on orthogonal projection matrices, whose proof is omitted. 

\begin{proposition} \label{prop:projMatrices}
	Let $L$ be any linear subspace of $ \C^n $.
\begin{enumerate}
	\item For any $ \mathbf x \in \C^n $,
	its \emph{orthogonal projection} $ \mathbf x^\perp _L $ onto $L$ 
	is the unique point of $L$
	that satisfies $( \mathbf {x - x}^\perp _L )^* ( \mathbf x^\perp _L - \mathbf y) = 0$ 
	for all $ \mathbf y \in L$.
	\item The mapping $ \C^n \owns \mathbf x \mapsto \mathbf x^\perp _L \in L$ is linear,
	so there exists a \emph{projection matrix} $ \boldsymbol \Pi _L $
	such that $ \boldsymbol \Pi _L \mathbf x = \mathbf x^\perp _L $ 
	for all $ \mathbf x \in \C^n $.  
	\item The projection matrix $ \boldsymbol \Pi _L $ 
	satisfies $ \boldsymbol \Pi _L ^2 = \boldsymbol \Pi _L = \boldsymbol \Pi _L^* $.
	\item If $ \tilde L$ is any linear subspace of $ \C^n $ that is orthogonal to $L$,
	then $ \boldsymbol \Pi _L + \boldsymbol \Pi _{ \tilde L} 
	= \boldsymbol \Pi _{L \oplus \tilde L} $.
	\item If $ \{ \mathbf b^k \}_{k = 1}^m $ is any orthonormal basis of $L$,
	then $ \boldsymbol \Pi _L = \sum _{k = 1}^m \boldsymbol \Pi _{[ \mathbf b^k ]} $
	where each $[ \mathbf b^k ]$ denotes the one-dimensional subspace
	spanned by the basis vector $ \mathbf b^k $.
	\item For all $ \mathbf x \in \C^n $, one has
 \( \mathbf x^* \boldsymbol \Pi _L \mathbf x 
 = \mathbf x^* \boldsymbol \Pi^* _L \boldsymbol \Pi _L \mathbf x 
 = \| \boldsymbol \Pi _L \mathbf x \|^2 \ge 0 \).  
\end{enumerate}
\end{proposition}

\subsection{One-Dimensional Orthogonal Projections}

The following proposition characterizes orthogonal projections 
onto one-dimensional linear subspaces of $ \C^n $.
 
\begin{proposition} \label{prop:oneDimProj}
	Given any $n$-vector $ \mathbf e$ in the unit sphere $ \mathbb S$ of $ \C^n $,
	the $n \times n$ matrix $ \mathbf P := \mathbf e \, \mathbf e^* $ is Hermitian
	and represents the orthogonal projection $ \Pi _{[ \mathbf e] }$ of $ \C^n $
	onto the one-dimensional subspace $[ \mathbf e] = \spn ( \{ \mathbf e \} )$.
\end{proposition}

\begin{proof}
	For each $ \mathbf e \in \mathbb S$,
	the $n \times n$ matrix $\mathbf P := \mathbf e \, \mathbf e^* $ satisfies:
\begin{enumerate}
	\item \( \mathbf P^* = ( \mathbf e \, \mathbf e^* )^* = ( \mathbf e^* )^* \mathbf e^* 
	= \mathbf e \, \mathbf e^* = \mathbf P \), so $ \mathbf P$ is Hermitian;
	\item Because matrix multiplication satisfies the associative law,
	and also $ \mathbf e^* \mathbf e = 1$, one has
 \[ \mathbf P^2 = ( \mathbf e \, \mathbf e^* )( \mathbf e \, \mathbf e^* )
	= \mathbf e ( \mathbf e^* \mathbf e ) \mathbf e^* = \mathbf e \, \mathbf e^* = \mathbf P \]
	This implies that $ \mathbf P$ is an orthogonal projection. 
	\item For all $ \mathbf x \in \C^n $, 
	because matrix multiplication satisfies the associative law,
	and $c := \mathbf e^* \mathbf x$ is a scalar in $\C$, one has
\[ \mathbf P \, \mathbf x = ( \mathbf e \, \mathbf e^* ) \mathbf x 
	= \mathbf e \, ( \mathbf e^* \mathbf x) = \mathbf e \, c 
	= c \, \mathbf e \in [ \mathbf e] \]
	It follows that $ \mathbf P$ is the orthogonal projection $ \Pi _{[ \mathbf e] }$ 
	of $ \C^n $ onto $[ \mathbf e] $. \qedhere
\end{enumerate}
\end{proof}

\subsection{Ortho-Partitions and Ortho-Measurability in $ \C^n $} \label{ss:orthoMeas}

\begin{proposition}
	Given any orthogonal decomposition $ \bigoplus _{d \in D} L_d $ of $ \C^n $:
	\begin{enumerate}
		\item any two different subsets 
		in the finite union \( \cup _{d \in D} \hat L_d \) are disjoint;
 		\item any family of vectors $ \{ \mathbf x_d \}_{d \in D} $ 
		with $ \mathbf x_d \in \hat L_d $ for each $d \in D$ is linearly independent, 
		and $ \# D \le n$.
	\end{enumerate}
\end{proposition}

\begin{proof}
	First, suppose that $ \mathbf x \in L_d \cap L_{d'} $ where $d \ne d'$.
	Because $ L_d $ and $ L_{d'} $ are orthogonal, one has $ \mathbf x^* \mathbf x = 0$,
	and so $ \mathbf {x = 0}$.
	It follows that $ \hat L_d $ and $ \hat L_{d'} $ are disjoint.
	
	Second, suppose that $ \mathbf 0 = \sum _{d \in D} \alpha _d \mathbf x_d $
	where $ \alpha _d \in \C$ and $ \mathbf x_d \in \hat L_d $ for all $d \in D$.
	Then, because orthogonality implies that $ \mathbf x_d ^* \mathbf x_{d'} = 0$ 
	whenever $d \ne d'$, one has
	\begin{align*}
 0 = \left( \sum\nolimits _{d \in D} \alpha _d \mathbf x_d \right) ^*
 		\left( \sum\nolimits _{d \in D} \alpha _d \mathbf x_d \right) 
	  &= \sum\nolimits _{d \in D} ( \bar \alpha _d \alpha _d  ) \mathbf x^*_d \mathbf x_d \\
	  &= \sum\nolimits _{d \in D} | \alpha _d | ^2 \mathbf x^*_d \mathbf x_d
	\end{align*}
	But $ \mathbf x_d ^* \mathbf x_d > 0$ for all $d \in D$, 
	so $ \alpha _d = 0$ for all $d \in D$.
	According to the standard definition, therefore,
	the family $ \{ \mathbf x_d \}_{d \in D} $ of vectors is linearly independent.
	Then, because the dimension $n$ is the maximum size 
	of any linearly independent subset of $ \C^n $,	one must have $ \# D \le n$.  
\end{proof}

\begin{definition} \label{def:orthoAlgebra}
	Let $ \bigoplus _{d \in D} L_d $ be any orthogonal decomposition of $ \C^n $. 
\begin{enumerate}
	\item The \emph{residual set} $ R^D $ of all vectors that are omitted 
	from the family $ \{ \hat L_d \mid d \in D \} $ of non-zero members 
	of the sets that make up the orthogonal decomposition 
	is defined by \( R^D := \C^n \setminus \cup _{d \in D} \hat L_d \).
	\item The \emph{ortho-partition} of $ \C^n $ induced by $ \bigoplus _{d \in D} L_d $ 
	is the partition 
\begin{equation} \label{eq:defFD}
	 \mathfrak P^D := \{ R^D \} \cup \left( \cup _{d \in D} \hat L_d \right)
\end{equation}
	\item The \emph{ortho-algebra} $ \Sigma ^D $ 
	is the $ \sigma $-algebra $ \sigma ( \mathfrak P^D )$ 
	of all subsets of $ \C^n $ generated by the cells in the ortho-partition $ \mathfrak P^D $.
\end{enumerate}
\end{definition} 

Given the ortho-partition $ \mathfrak P^D $ of $ \C^n $ defined by \eqref{eq:defFD}, 
let $ \mathfrak F^D $ denote the family of all collections $ \mathfrak E^D $ 
of cells $E$ that make up the partition $ \mathfrak P^D $.
We state the following Lemma without proof:   

\begin{lemma}
	The $ \sigma $-algebra $ \sigma ( \mathfrak P^D )$ generated by $ \mathfrak P^D $ 
	is the collection of all subsets of $ \C^n $ which, 
	for some collection $ \mathfrak E^D \in \mathfrak F^D $, take the form 
	of the union $ \cup _{E \in \mathfrak E^D } E$ of the cells in $ \mathfrak E^D $.
\end{lemma}

\subsection{The Spectral Decomposition of a Hermitian Matrix}

\begin{definition} \label{def:eigenpairs}
	Given any $n \times n$ Hermitian matrix $ \mathbf A$ 
	and any of its eigenvalues $ \lambda \in s^{ \mathbf A} $, let: 
\begin{enumerate}
	\item \( E_\lambda := \{ \mathbf x \in \C^n \setminus \{ \mathbf 0\} 
	\mid \mathbf A \mathbf x = \lambda \mathbf x \} \)
	be the corresponding non-empty set of its eigenvectors;
	\item \( L_\lambda := E_\lambda \cup \{ \mathbf 0\} \) 
	be the corresponding (linear) \emph{eigenspace}.
\end{enumerate}
\end{definition}

Evidently Definition \ref{def:eigenpairs} implies that $ E_\lambda = \hat L_\lambda $
for all $ \lambda \in s^{ \mathbf A} $.

Our first main theorem in Section \ref{ss:orthMeasFns} will make use of the concept
of numerically identified orthogonal decompositions, which are defined as follows:

\begin{definition} \label{def:numLabelledOrthDecomp}
	The orthogonal decomposition $ \bigoplus _{d \in D} L_d $ of $ \C^n $ 
	is \emph{numerically identified} just in case 
	there exist a finite subset $ \Lambda \subset \R$ of numerical identifiers
	and a bijection $D \owns d \longleftrightarrow \lambda _d \in \Lambda $. 
\end{definition}

Our next result, which we also state without proof, 
is a version of the standard spectral theorem 
that applies in the finite-dimensional space $ \C^n $.

\begin{proposition} \label{prop:spectralThm}
	Any $n \times n$ Hermitian matrix $ \mathbf A$:
\begin{enumerate}
	\item induces a numerically identified orthogonal decomposition 
 \( \bigoplus\nolimits _{ \lambda \in s^{ \mathbf A} } L_\lambda \) of $ \C^n $
	into linear subspaces $ L_\lambda = E_\lambda \cup \{ \mathbf 0\} $
	that correspond to the eigenspaces $ E_\lambda $ of $ \mathbf A$,
	one for each $ \lambda \in s^{ \mathbf A} $ in the spectrum of $ \mathbf A$;
	\item has a \emph{spectral decomposition} into the eigenvalue-weighted linear combination 
\begin{equation} \label{eq:spectralDecomp}
	\mathbf A = \sum\nolimits _{ \lambda \in s^{ \mathbf A} } 
	\lambda \, \boldsymbol \Pi _{ L_\lambda }
\end{equation}
	of orthogonal projections $ \boldsymbol \Pi _{ L_\lambda }$ 
	onto the corresponding mutually orthogonal eigenspaces $ L_\lambda $.
\end{enumerate}
\end{proposition}

\section{Quantum Observables as Measurable Functions} \label{s:firstBijection}
\subsection{From a Hermitian Matrix to Its Eigen-Pairing}

Our construction of an ortho-measurable function 
that corresponds to a quantum observable in the form of a Hermitian matrix $ \mathbf A$
will involve an extension of the real line $\R$ that adds an extra element $*$.
Following terminology which is common in computer science,
the extra element $*$ could be read as ``not a number'', often denoted by ``NaN''.
The element $*$ is used in defining an eigenpairing 
as a mapping  \( \C^n \owns \mathbf x \mapsto f^{ \mathbf A} ( \mathbf x) \) 
whose value is only a well-defined real number 
in case $ \mathbf x$ is an eigenvector of~$ \mathbf A$.

\begin{definition} \label{defn:eigenPairing}
Given any $n \times n$ Hermitian matrix $ \mathbf A$ with spectral decomposition 
 \( \mathbf A 
 = \sum _{ \lambda \in s^{ \mathbf A} } \lambda \, \boldsymbol \Pi _{ L_\lambda } \), 
as in \eqref{eq:spectralDecomp}, define:
 \begin{enumerate}
	\item for each eigenvalue $ \lambda \in s^{ \mathbf A}$
	and unique associated linear eigenspace $ L_\lambda $, 
	the set $ E_\lambda := L_\lambda \setminus \{ \mathbf 0 \} $ of corresponding eigenvectors, 
	as well as the \emph{indicator function} 
 \( \C^n \owns \mathbf x \mapsto 1_{ E_\lambda } ( \mathbf x) \in \{0, 1\} \) 
which is defined so that
 \( 1_{ E_\lambda } ( \mathbf x) = 1 \Longleftrightarrow \mathbf x \in E_\lambda \);
		\item the \emph{residual set}
 \( R^{ \mathbf A} := \C^n \setminus \cup _{ \lambda \in s^{ \mathbf A} } E_\lambda \) 
 	of $n$-vectors (including $ \mathbf 0$)
	that are not eigenvectors of $ \mathbf A$, for any of its eigenvalues;  
	\item the \emph{one-point extension} $ \R \cup \{*\} $ of the real line, 
	where $ * \not\in \R $;
	\item the \emph{eigen-pairing} of $ \mathbf A$ as the map
 \( \C^n \owns \mathbf x \mapsto f^{ \mathbf A} ( \mathbf x) \in \R \cup \{*\} \)
 	 defined by 
\begin{equation} \label{eq:defEigenPairs}
	f^{ \mathbf A} ( \mathbf x) := \begin{cases}
 	  \sum _{ \lambda \in s^{ \mathbf A}} \lambda \, 1_{ E_\lambda } ( \mathbf x)
  	& \text{if $ \mathbf x \in \C^n \setminus R^{ \mathbf A} $}; \\
  * & \text{if $ \mathbf x \in R^{ \mathbf A} $}.
\end{cases}
\end{equation}
\end{enumerate}
\end{definition}

By Proposition \ref{orthogSpaces}, 
the eigenspaces $ E_\lambda = L_\lambda \setminus \{ \mathbf 0 \}$ 
of any $n \times n$ Hermitian matrix $ \mathbf A$ are orthogonal, 
and so disjoint, for different values of~$ \lambda $.
Hence \eqref{eq:defEigenPairs} implies that 
\begin{equation} \label{eq:keyEq}
	f^{ \mathbf A} ( \mathbf x) = \lambda \Longleftrightarrow \mathbf x \in E_\lambda 
  \Longleftrightarrow \mathbf x \ne \mathbf 0 \ \text{and} \ \mathbf {Ax} = \lambda \mathbf x 
\end{equation}
It follows that the mapping $ \mathbf x \mapsto f^{ \mathbf A} ( \mathbf x) $ 
pairs each $ \mathbf x \in \C^n $ 
with the extra point $*$ in case $ \mathbf x$ belongs to the residual set $ R^{ \mathbf A} $,
but with the eigenvalue $ \lambda $ in case $ \mathbf x \in \C^n \setminus R^{ \mathbf A} $
is an eigenvector in the eigenspace $ E_\lambda $.
In particular, the following result holds.

\begin{lemma} \label{lem:range}
	The range $ f^{ \mathbf A} ( \C^n )$ 
		of the map defined by \eqref{eq:defEigenPairs} equals 
		the finite set $ s^{ \mathbf A} \cup \{ * \} \subset \R \cup \{ * \}$.
\end{lemma}

\begin{remark}
	An alternative definition of the function $ \mathbf x \mapsto f^{ \mathbf A} ( \mathbf x) $
	would exclude from its domain the residual set $ R^{ \mathbf A} $ specified 
	in part 2 of Definition \ref{defn:eigenPairing}.
	It seems more in the spirit of probability theory, however, 
	to define $ \C^n \owns \mathbf x \mapsto f^{ \mathbf A} ( \mathbf x) $ on all of $ \C^n $, 
	but then to attach probability zero to the residual set.
\end{remark}

\subsection{Making the Eigen-Pairing Measurable}

We begin with a preliminary lemma.

\begin{lemma}
	Consider the smallest $ \sigma $-algebra $ \sigma ( \mathcal B \cup \{*\} )$ 
	on the co-domain $ \R \cup \{*\} $ that includes the singleton set $ \{*\} $
	in addition to all sets in the Borel $ \sigma $-algebra $ \mathcal B$ on~$ \R $.
	Then $ \sigma ( \mathcal B \cup \{*\} )$ is the union
\begin{equation} \label{eq:defSigma}
	\Sigma := \mathcal B \cup \mathcal B^*
\end{equation}
	of the Borel $ \sigma $-algebra $ \mathcal B$ on $\R$ 
	with the family $ \mathcal B^* := \{ B \cup \{*\} \mid B \in \mathcal B \} $. 
\end{lemma}

\begin{proof}
	First, any $ \sigma $-algebra that includes all the Borel sets $B \in \mathcal B$ 
	as well as the singleton set $ \{*\} $ must obviously include every set 
	in the family $ \Sigma $ defined by \eqref{eq:defSigma}.
	It remains only to prove that $ \Sigma $ is itself a $ \sigma $-algebra.
	
	Evidently the family $ \Sigma $ includes the whole co-domain $ \R \cup \{ * \} $ 
	as well as, for each Borel set $B \in \mathcal B$, the respective complements 
	of the two sets $B$ and $B \cup \{*\} $ in $ \R \cup \{ * \} $, which are
\begin{equation}
\begin{array} {rcl}
	( \R \cup \{ * \} ) \setminus B &=& ( \R \setminus B) \cup \{ * \} \\
	\text{and} \quad ( \R \cup \{ * \} ) \setminus (B \cup \{*\} ) &=& \R \setminus B
\end{array}
\end{equation}
	Because $( \R \setminus B) \cup \{ * \} \in \mathcal B^* $
	and $( \R \cup \{ * \} ) \setminus (B \cup \{*\} ) \in \mathcal B$,
	it follows that the complement of any set in $ \Sigma $ also belongs to $ \Sigma $.
	
	Finally, consider the union of any countable family $ \mathcal F$ of sets in $ \Sigma $.
	This union is: either (i), in case no set in $ \mathcal F$ has $*$ as a member, 
	a union $ \cup _{k \in K} B_k $ of a countable family $ \{ B_k \}_{k \in K} $
	of sets in $ \mathcal B$;
	or (ii), in case at least one set in $ \mathcal F$ has $*$ as a member,
	a union $ \left ( \cup _{k \in K} B_k \right) \cup \{*\} $.
	In either case, the countable union $ \cup _{F \in \mathcal F} F$ 
	belongs either to $ \mathcal B$ or to $ \mathcal B^* $, and so to $ \Sigma $.
	
	This completes the confirmation that $ \Sigma $ defined by \eqref{eq:defSigma}
	is a $ \sigma $-algebra on $ \R \cup \{ * \} $.
\end{proof}

\begin{theorem}
	Let $ \Sigma ^{ \mathbf A} $ denote the ortho-algebra on $ \C^n $ 
	which results from applying part 3 of Definition \ref{def:orthoAlgebra}
	to the orthogonal decomposition $ \bigoplus _{ \lambda \in s^{ \mathbf A} } L_\lambda $.
	Then the eigen-pairing 
 \( \C^n \owns \mathbf x \mapsto f^{ \mathbf A} ( \mathbf x) \in \R \cup \{*\} \)
	defined by \eqref{eq:defEigenPairs} yields a measurable function
	from the measurable space $( \C^n, \Sigma ^{ \mathbf A} )$ to the measurable space 
 \( ( \R \cup \{*\}, \sigma ( \mathcal B \cup \{*\} ) = ( \R \cup \{*\}, \Sigma ) \).
\end{theorem}

\begin{proof}
	Given any Borel set $B \subset \R$, it follows from \eqref{eq:keyEq} that
 \( ( f^{ \mathbf A} )^{-1} (B) = \cup _{ \lambda \in B \cap s^{ \mathbf A} } E_\lambda \)
 	and also that $( f^{ \mathbf A} )^{-1} (B \cup \{ * \}) $ 
	is the union $( f^{ \mathbf A} )^{-1} (B) \cup R^{ \mathbf A} $
	of $( f^{ \mathbf A} )^{-1} (B) $ with the residual set $ R^{ \mathbf A} $.
	Then both the sets $( f^{ \mathbf A} )^{-1} (B) $ 
	and $( f^{ \mathbf A} )^{-1} (B) \cup R^{ \mathbf A} $
	are obviously $ \Sigma ^{ \mathbf A} $-measurable, 
	as unions of finitely many $ \Sigma ^{ \mathbf A} $-measurable subsets of $ \C^n $.
\end{proof}

\subsection{Hermitian Matrices as Ortho-Measurable Functions} \label{ss:orthMeasFns}

This section is devoted to the first main theorem of the paper.
It involves functions meeting the following definition:

\begin{definition} \label{def:orthMeasFn}
	A function \( \C^n \owns \mathbf x \mapsto f( \mathbf x) \in \R \cup \{*\} \) 
	is \emph{ortho-measurable} just in case there exists 
	an orthogonal decomposition $ \bigoplus _{d \in D} L_d $ of $ \C^n $ for which:
	\begin{enumerate}
		\item for all $ \lambda \in \R \cap f( \C^n )$, there exists
		a unique $d \in D$ and so a unique space $ L_d $ of the decomposition such that
		 $ f^{-1} ( \{ \lambda \} ) = L_d \setminus \{ \mathbf 0 \}$;
		\item $ f^{-1} ( \{*\} )$ is the residual set 
 \( R = \C^n \setminus \cup _{d \in D} ( L_d \setminus \{ \mathbf 0 \} ) \).
	\end{enumerate}
\end{definition}

With this additional definition, the theorem can be stated as follows:

\begin{theorem} \label{thm:threeBijections}
	There exist bijections between each pair of the three sets of:
\begin{enumerate}
	\item quantum observables in the form of $n \times n$ Hermitian matrices $ \mathbf A$;
	 \item numerically identified 
	 orthogonal decompositions $ \bigoplus _{ \lambda \in \Lambda } L_\lambda $ of $ \C^n $; 
	\item ortho-measurable functions
	 \( \C^n \owns \mathbf x \mapsto f( \mathbf x) \in \R \cup \{*\} \)
	 satisfying $ f^{-1} ( \{ \lambda \} ) = E_\lambda $ for all $ \lambda \in \Lambda $.
\end{enumerate}
\end{theorem}

\begin{proof}
	First, given any $n \times n$ Hermitian matrix $ \mathbf A$, 
	the spectral theorem \ref{prop:spectralThm} implies the existence 
	of the unique numerically identified orthogonal decomposition 
	 \( \bigoplus\nolimits _{ \lambda \in \Lambda } L_\lambda \) of $ \C^n $,
	 where $ \Lambda $ is the finite spectrum $ s^{ \mathbf A}$ of $ \mathbf A$. 

	Second, let \( \C^n \owns \mathbf x \mapsto f( \mathbf x) \in \R \cup \{*\} \) 
	be any ortho-measurable function.
	Now relabel the sets $ L_d $ 
	in the orthogonal decomposition $ \bigoplus _{d \in D} L_d $ of $ \C^n $ 
	of Definition \ref{def:orthMeasFn} according to the function values $ \lambda $
	in the real part $ \Lambda = \R \cap f( \C^n )$ of the range of $f$, which is a finite set.
	For different values of~$ \lambda $, the corresponding spaces $ L_\lambda $ are orthogonal.
	So the resulting 
	labelled orthogonal decomposition $ \bigoplus _{ \lambda \in \Lambda } L_\lambda $
	is numerically identified.
	Then the decomposition $ \bigoplus _{ \lambda \in \Lambda } L_\lambda $ serves 
	as a parameter of the function $f$ defined by 
\begin{equation} \label{eq:defOMeasFn}
	f \left( \bigoplus\nolimits _{ \lambda \in \Lambda } L_\lambda ; \mathbf x \right) 
	:= \begin{cases} 
	\lambda & \text{if $ \mathbf x \in E_\lambda $}; \\
 	* & \text{if $ \mathbf x \in \C^n \setminus \cup _{ \lambda \in \Lambda } E_\lambda $}.
\end{cases}
\end{equation}
	
	Next, consider the intersection of the graph of this function, 
	which is a subset of the Cartesian product $ \C^n \times ( \R \cup \{*\} )$, 
	with the product set $( \C^n \setminus \{ \mathbf 0\} ) \times \R$. 
	This intersection is the \emph{restricted graph} made up of the finite union 
	of pairwise disjoint sets given by
\begin{equation} \label{eq:graphOMeasFn}
	\Gamma \left( \bigoplus\nolimits _{ \lambda \in \Lambda } L_\lambda \right) 
	:= \bigcup\nolimits _{ \lambda \in \Lambda } \left( E_\lambda \times \{ \lambda \} \right)
\end{equation}
	
	Finally, recalling the spectral decomposition \eqref{eq:spectralDecomp} 
	and putting $ \Lambda = s^{ \mathbf A} $, we have the following chain of bijections
\begin{equation} \label{eq:bijections}
	\mathbf A = \sum\nolimits _{ \lambda \in s^{ \mathbf A} } 
	\lambda \, \boldsymbol \Pi _{ L_\lambda } 
	\longleftrightarrow \bigoplus\nolimits _{ \lambda \in \Lambda } L_\lambda 
	\longleftrightarrow \Gamma \left( \bigoplus\nolimits _{ \lambda \in \Lambda } L_\lambda \right)
\end{equation}
	This chain makes the result evident.
\end{proof}

\subsection{A Contextual Multi-Measurable Space} \label{ss:contextualMultMeas}

Let $ \mathcal D$ denote the family of all orthogonal decompositions of $ \C^n $.
Then, for each orthogonal decomposition $D \in \mathcal D$
and associated ortho-algebra $ \Sigma ^D $,
the pair $( \C^n, \Sigma ^D )$ is a measurable space
that depends on the orthogonal decomposition $D$,
regarded as a \emph{context}.
So, according to the definitions in Hammond (2025), 
the pair $( \C^n, ( \Sigma ^D )_{D \in \mathcal D}) $
with the complete family of all ortho-algebras is a \emph{multi-measurable space}.

\section{Wave Vectors as Pre-Probabilities} \label{s:waveVectors}
\subsection{Preliminary Definitions}

In quantum theory with the quantum state space $ \C^n $,
according to usual terminology, a ``normalized wave function'' is a mapping 
 \( T \owns t \mapsto \psi (t) \in \Sph \) from a time interval $T \subset \R$ 
to the unit sphere $ \Sph $. 
Since the current focus is on events at one fixed time, we use the following definition.

\begin{definition}
	A \emph{wave vector} is an element of the Hilbert space $ \C^n \setminus \{ \mathbf 0\}$. 
	The wave vector $ \boldsymbol \psi \in \C^n \setminus \{ \mathbf 0\} $ is \emph{normalized}
	just in case it is an element of the unit sphere $ \Sph $ in $ \C^n $.
\end{definition}

\begin{definition}
	Let $ \mathcal B = ( \mathbf b^j )_{j = 1}^n $ be any orthonormal basis of $ \C^n $.
\begin{enumerate}	
	\item Given any $j \in \N_n $ and any corresponding $ \mathbf b^j \in \mathcal B$, let:
\begin{itemize}
	\item $ L_j := [ \mathbf b^j ]$ denote the one-dimensional linear space
	spanned by the basis vector $ \mathbf b^j $;
	\item $ \boldsymbol \Pi _{ L_j }$ denote 
	the $n \times n$ Hermitian matrix $ \mathbf b^j ( \mathbf b^j )^* $ which, 
	by Proposition \ref{prop:oneDimProj}, 
	represents the orthogonal projection mapping onto $ L_j $.   
\end{itemize}
	\item Let $ D^{ \mathcal B} := \bigoplus _{j = 1}^n L_j $
	denote the associated \emph{basic} orthogonal decomposition of $ \C^n $.
	\item Let $ \mathfrak P^{ \mathcal B} $ denote the associated \emph{basic} ortho-partition
	whose cells are, as in Definition \ref{def:orthoAlgebra}, 
	the $n$ sets $ L_j \setminus \{ \mathbf 0 \}$, together with the residual set 
 \( R^{ \mathcal B} = \C^n \setminus \cup _{j = 1}^n ( L_j \setminus \{ \mathbf 0 \}) \).
	\item Let $ \Sigma ^{ \mathcal B} = \sigma ( \mathfrak P^{ \mathcal B} )$ 
	denote the associated \emph{basic} ortho-algebra generated by the cells 
	which constitute the basic ortho-partition $ \mathfrak P^{ \mathcal B} $.
\end{enumerate}
\end{definition}

\subsection{Wave Vectors as Pre-Probabilities: Special Case}

Consider the special case of a basic ortho-algebra $ \Sigma ^{ \mathcal B} $
associated with an orthonormal basis $ \mathcal B = ( \mathbf b^j )_{j = 1}^n $ of $ \C^n $.
Any normalized wave vector $ \boldsymbol \psi \in \Sph $ can be used 
to construct a probability measure $ \mathbb P ^{ \mathcal B}_{ \boldsymbol \psi }$ 
on the measurable space $( \C^n, \Sigma ^{ \mathcal B} )$,
according to \emph{Born's rule}, treating $ \boldsymbol \psi $ as a parameter.
This rule requires that, for each $j \in \N_n $ and $ L_j = [ \mathbf b^j ]$,
the probability 
 \( \mathbb P^{ \mathcal B}_{ \boldsymbol \psi } ( L_j \setminus \{ \mathbf 0 \} ) \)
of the basic set $ L_j \setminus \{ \mathbf 0 \} \in \Sigma ^{ \mathcal B} $
is given by \( \boldsymbol \psi ^* \boldsymbol \Pi _{ L_j } \boldsymbol \psi \).%
\footnote{Note that when the vector $ \boldsymbol \psi $ is not normalized,
this formula is replaced by the \emph{Rayleigh quotient} 
 \( \boldsymbol \psi ^* \boldsymbol \Pi _{ L_j } \boldsymbol \psi 
 / | \boldsymbol \psi |^2 \).}
But $ \boldsymbol \Pi _{ L_j } = \mathbf b^j ( \mathbf b^j )^* $ in this special case, so
\begin{equation} \label{eq:BornRule}
	\mathbb P^{ \mathcal B}_{ \boldsymbol \psi } ( L_j \setminus \{ \mathbf 0 \} )
	= \boldsymbol \psi ^* \mathbf b^j ( \mathbf b^j )^* \boldsymbol \psi
	= | \boldsymbol \psi ^* \mathbf b^j |^2
\end{equation}

In the even more special case of the \emph{canonical orthonormal basis} $ \mathcal B$,
which consists of the $n$ columns of the $n \times n$ identity matrix, 
one has $ \boldsymbol \psi ^* \mathbf b^j = \psi _j $.
Then Born's rule formula \eqref{eq:BornRule} evidently reduces to 
 \( \mathbb P ^{ \mathcal B}_{ \boldsymbol \psi } ( L_j \setminus \{ \mathbf 0 \} ) 
 = | \psi _j |^2 \),
which is the \emph{squared modulus rule} for the probability
of each basic set $ L_j \setminus \{ \mathbf 0 \}$ 
in the orthogonal decomposition $ \bigoplus _{j = 1}^n ( L_j \setminus \{ \mathbf 0 \} )$.
Because normalization implies that the components of the wave vector
satisfy $ \sum _{j = 1}^n | \psi _j |^2 = 1$, these probabilities do sum to one. 

\subsection{Wave Vectors as Pre-Probabilities: General Case}

Let $ \bigoplus _{d \in D} L_d $ be any orthogonal decomposition of $ \C^n $,
with associated ortho-partition $ \mathfrak P^D $ and ortho-algebra $ \Sigma ^D $.
Born's rule requires the probability 
of each non-residual cell $ L_d \setminus \{ \mathbf 0 \} $ 
in the ortho-partition $ \mathfrak P^D $ to satisfy
\begin{equation} \label{eq:Born}
	 \mathbb P^D_{ \boldsymbol \psi } ( L_d \setminus \{ \mathbf 0 \} )
	 = \boldsymbol \psi ^* \boldsymbol \Pi _{ L_d } \boldsymbol \psi
\end{equation}
where, unlike in \eqref{eq:BornRule}, 
the orthogonal projection represented by the matrix $ \boldsymbol \Pi _{ L_d }$
may be onto a space $ L_d $ whose dimension exceeds one.

Note that for an orthogonal decomposition,
the sum $ \sum _{d \in D} \boldsymbol \Pi _{ L_d }$ of all the orthogonal projections 
onto the component subspaces $ L_d $ equals the identity matrix $ \mathbf I$.
It follows that
\begin{align}
   \sum\nolimits _{d \in D} 
   \mathbb P^D_{ \boldsymbol \psi } ( L_d \setminus \{ \mathbf 0 \} ) 
 &= \sum\nolimits _{d \in D} \boldsymbol \psi ^* \boldsymbol \Pi _{ L_d } \boldsymbol \psi 
 \notag \\
 &= \boldsymbol \psi ^* \left( \sum\nolimits _{d \in D} \boldsymbol \Pi _{ L_d } \right)
 	 \boldsymbol \psi 
 = \boldsymbol \psi ^* \mathbf I \boldsymbol \psi = \boldsymbol \psi ^* \boldsymbol \psi 
 = 1 \label{eq:normProbs}
\end{align}
Also $ \mathbb P^D_{ \boldsymbol \psi } ( R^D ) = 0$ for the residual set
 \( R^D = \C^n \setminus \cup _{d \in D} ( L_d \setminus \{ \mathbf 0\} ) \).

\subsection{The CDF of an Observable for a Given Wave Vector}

Consider any normalized wave vector $ \boldsymbol \psi \in \Sph $, 
along with any quantum observable in the form of a Hermitian matrix $ \mathbf A$
whose spectral decomposition is the eigenvalue-weighted sum 
 \( \mathbf A = \sum _{ \lambda \in s^{ \mathbf A} }
 	 \lambda \, \boldsymbol \Pi _{ L^{ \mathbf A} _\lambda } \) of the family 
 \( \left\{ \boldsymbol \Pi _{ L^{ \mathbf A} _\lambda } 
\mid \lambda \in s^{ \mathbf A} \right\} \) of projection matrices.
Then the pair $( \boldsymbol \psi, \mathbf A) $ induces a random variable 
 \( \C^n \owns \mathbf x \mapsto f^{ \mathbf A} ( \mathbf x) \in \R \cup \{ * \} \)
on the ortho-measurable space $( \C^n, \Sigma ^{ \mathbf A}) $
with a \emph{cumulative distribution function} (or CDF)
\( \R \owns r \mapsto F^{ \mathbf A} _{ \boldsymbol \psi } (r) \in [0, 1] \) which, as usual,
specifies for each $r \in \R$ the probability that the random variable 
satisfies $ f^{ \mathbf A} ( \mathbf x) \le r$.
This CDF takes the form
\begin{equation} \label{eq:CDF}
	F^{ \mathbf A} _{ \boldsymbol \psi } (r) = \mathbb P^{ \mathbf A} _{ \boldsymbol \psi } 
	\left( ( f^{ \mathbf A} )^{-1} ( - \infty, r] \right)
	= \sum\nolimits _{ \lambda \in s^{ \mathbf A} } 1_{ \lambda \le r} ( \lambda ) \,
	\boldsymbol \psi ^* \boldsymbol \Pi _{ L^{ \mathbf A} _\lambda } \boldsymbol \psi
\end{equation}
Because $ \bigoplus _{ \lambda \in s^{ \mathbf A} } L^{ \mathbf A} _\lambda = \C^n $,
it follows from \eqref{eq:CDF} and then \eqref{eq:normProbs} that
\begin{equation} \label{eq:CDFtop}
	F^{ \mathbf A} _{ \boldsymbol \psi } ( + \infty ) 
	= \mathbb P^{ \mathbf A} _{ \boldsymbol \psi } 
	\left( ( f^{ \mathbf A} )^{-1} ( \R ) \right)
	= \sum\nolimits _{ \lambda \in s^{ \mathbf A} }
	\boldsymbol \psi ^* \boldsymbol \Pi _{ L^{ \mathbf A} _\lambda } \boldsymbol \psi
	= 1
\end{equation}
This, of course, implies that the CDF gives probability zero 
to the residual event that $ f^{ \mathbf A} ( \mathbf x) = *$.

Because of the spectral decomposition 
 \( \mathbf A = \sum _{ \lambda \in s^{ \mathbf A} }
 	 \lambda \, \boldsymbol \Pi _{ L^{ \mathbf A} _\lambda } \) of $ \mathbf A$ 
and linearity, the \emph{expectation} of the induced random variable $ f^{ \mathbf A} $ is
\[ \mathbb E_{ \boldsymbol \psi } f^{ \mathbf A} 
 	= \sum\nolimits _{ \lambda \in s^{ \mathbf A} } \lambda
	\boldsymbol \psi ^* \boldsymbol \Pi _{ L^{ \mathbf A} _\lambda } \boldsymbol \psi
	= \boldsymbol \psi ^* \mathbf A \boldsymbol \psi \]

\subsection{A Multi-Probability Space for a Given Wave Vector} \label{ss:multProbSpace}

Given any fixed normalized wave vector $ \boldsymbol \psi \in \Sph $,
we can now use the family $( \mathbb P^D_{ \boldsymbol \psi } )_{D \in \mathcal D} $ 
of probability measures we have just defined in order to extend:
\begin{itemize}
	\item the previous \emph{multi-measurable} space
	 $( \C^n, ( \Sigma ^D )_{D \in \mathcal D}) $
	defined in Section \ref{ss:contextualMultMeas},
    with a complete family of ortho-algebras $ \Sigma ^D $,
    one for each orthogonal decomposition $D \in \mathcal D$;
\item into a \emph{multi-probability} space 
 \( ( \C^n, ( \Sigma ^D, \mathbb P^D_{ \boldsymbol \psi } )_{D \in \mathcal D}) \),
    with a complete family of \emph{contextual probability} spaces
 \( ( \C^n, \Sigma ^D, \mathbb P^D_{ \boldsymbol \psi }) \),
    one for each orthogonal decomposition $D \in \mathcal D$.
\end{itemize}

\section{Density Matrices and the Trace Formula} \label{s:densityMatrices}
\subsection{Key Properties of the Trace of a Matrix}

Recall that the \emph{trace} $ \tr \mathbf A$ of any $n \times n$ matrix $ \mathbf A$ 
is the sum $ \sum _{j = 1}^n a_{jj} $ of the elements on its principal diagonal. 
In case $ \mathbf A$ is Hermitian, these diagonal elements are all real, 
and so therefore is the trace.

\begin{lemma} \label{propTrace}
	Let $ \mathbf A = ( a_{ij} )_{n \times n} $ and $ \mathbf B = ( b_{ji} )_{n \times n} $
	be complex $n \times n$ matrices.
	Then: (i) $ \tr ( \mathbf {AB} ) = \tr ( \mathbf {BA} )$;
	(ii) if $ \mathbf B^{-1} $ exists,
	then $ \tr ( \mathbf B^{-1} \mathbf {AB} ) = \tr \mathbf A$.   
\end{lemma}

\begin{proof} 
For part (i), one has
\[ \tr ( \mathbf {AB} ) = \sum\nolimits _{i = 1}^n \sum\nolimits _{j = 1}^n a_{ij} b_{ji}
					= \sum\nolimits _{j = 1}^n \sum\nolimits _{i = 1}^n b_{ji} a_{ij} 
					= \tr ( \mathbf {BA} ) \]
For part (ii), put $ \mathbf U = \mathbf B^{-1} $ and $ \mathbf V = \mathbf {AB} $. 
Because part (i) implies that $ \tr ( \mathbf {UV} ) = \tr ( \mathbf {VU} )$,
one has $ \tr ( \mathbf B^{-1} \mathbf {AB} ) = \tr ( \mathbf {ABB}^{-1} ) = \tr \mathbf A$.
\end{proof}

\begin{proposition} \label{prop:trace}
	Suppose that $ \mathbf A$ is an $n \times n$ Hermitian matrix 
	whose spectral decomposition, as in \eqref{eq:spectralDecomp},
	is \( \mathbf A = \sum\nolimits _{ \lambda \in s^{ \mathbf A} } 
			\lambda \, \boldsymbol \Pi _{ L_\lambda } \). 
	For each $ \lambda \in s^{ \mathbf A} $, let $ m_\lambda $ denote 
	the dimension of the linear space $ L_\lambda $.
	Then 
\begin{equation} \label{eq:traceFormula}
	\tr \mathbf A = \sum\nolimits _{ \lambda \in s^{ \mathbf A} } \lambda \, m_\lambda
\end{equation}
\end{proposition}

\begin{proof} 
	A well-known property of any Hermitian matrix $ \mathbf A$ is that it can be diagonalized, 
	meaning that there exist an $n \times n$ unitary matrix $ \mathbf U$ 
	and an $n \times n$ diagonal matrix $ \mathbf D$ 
	such that $ \mathbf U \mathbf A \mathbf U^{-1} = \mathbf D$,
	and so $ \mathbf U^{-1} \mathbf D \mathbf U = \mathbf A$. 
	Moreover, the diagonal entries of $ \mathbf D$ are the eigenvalues of $ \mathbf A$.
	To allow for repeated eigenvalues, note that for each $ \lambda \in s^{ \mathbf A} $, 
	the dimension $ m_\lambda $ of $ L_\lambda $ equals the number of times that $ \lambda $ 
	appears on the diagonal of $ \mathbf D$.
	So it follows from part (ii) of Lemma \ref{propTrace} that
\[ \tr \mathbf A = \tr ( \mathbf U^{-1} \mathbf D \mathbf U ) 
	= \tr \mathbf D = \sum\nolimits _{ \lambda \in s^{ \mathbf A} } \lambda \, m_\lambda 
	\qedhere \]
\end{proof}

The right-hand side of \eqref{eq:traceFormula} can be described
as the \emph{dimensionally} or \emph{multiplicity weighted sum} of the eigenvalues.

\begin{proposition} \label{prop:traceProj}
	Let $L$ be any linear subspace of $ \C^n $ whose dimension is $ m_L $.
	Then the trace $ \tr \boldsymbol \Pi _L $ of the orthogonal projection $ \Pi _L $
	onto $L$ is $ m_L $.
\end{proposition}

\begin{proof}
	Consider any diagonalization $ \mathbf D = \mathbf U \boldsymbol \Pi _L \mathbf U^{-1} $ 
	of $ \boldsymbol \Pi _L $, where $ \mathbf U$ is a unitary matrix.
	It is easy to see that if $ \mathbf x \ne \mathbf 0$ 
	and $ \boldsymbol \Pi _L \mathbf x = \lambda \mathbf x$, then:
	either (i) $ \lambda = 1$ and $ \boldsymbol \Pi _L \mathbf x = \mathbf x$,
	implying that $ \mathbf x \in L$;
	or (ii) $ \lambda = 0$ and $ \boldsymbol \Pi _L \mathbf x = \mathbf 0$,
	implying that $ \mathbf x \in L^\perp $, the orthogonal complement of $L$
	that satisfies $L \bigoplus L^\perp = \C^n $.
	It follows that $ \mathbf D$ has $ m_L $ diagonal elements equal to 1,
	with the remaining $n - m_L $ diagonal elements all equal to~0.
	But then Lemma \ref{propTrace} implies that 
 \( \tr \boldsymbol \Pi _L = \tr ( \mathbf U^{-1} \mathbf D \mathbf U) = \tr \mathbf D = m_L \).
\end{proof}

\subsection{Density Matrices and Their Spectrum} \label{ss:densityMatrices}

The following is yet another standard definition.

\begin{definition}
	An $n \times n$ Hermitian matrix $ \boldsymbol \rho $ is:
	\begin{enumerate}
		\item \emph{positive semi-definite} just in case 
		\( \boldsymbol \psi ^* \boldsymbol \rho \, \boldsymbol \psi \ge 0 \)
		for all $ \boldsymbol \psi \in \C^n $;%
		\footnote{Most physicists and some mathematicians say 
		that such a matrix~$ \boldsymbol \rho $ is \emph{positive}.}
		\item a \emph{density matrix} just in case it is positive semi-definite 
		and $ \tr \boldsymbol \rho = 1$.
	\end{enumerate}
\end{definition}

\begin{definition}
	Given any $n \times n$ density matrix $ \boldsymbol \rho $ 
	and eigenvalue $ \lambda \in s^{ \boldsymbol \rho }$, 
	let $ m_\lambda \in \N$ denote the dimension of the corresponding eigenspace $ L_\lambda $.
\end{definition}

The following result gives a well known characterization.

\begin{proposition} \label{prop:densityMatrix}
	The $n \times n$ Hermitian matrix $ \boldsymbol \rho $ is a density matrix if and only if:	
	(i) all its eigenvalues $ \lambda \in s^{ \boldsymbol \rho }$ satisfy $ \lambda \ge 0$; 
	(ii) together the eigenvalues also satisfy the unit trace equation
\begin{equation} \label{eq:traceOne}
	\sum\nolimits _{ \lambda \in s^{ \boldsymbol \rho }}  m_\lambda \lambda = 1
\end{equation}
\end{proposition}

\begin{proof}
	Given the $n \times n$ Hermitian matrix $ \boldsymbol \rho $, 
	let $ \mathbf D = \mathbf U \boldsymbol \rho \mathbf U^{-1} $ be any diagonalization.
\begin{enumerate}
	\item It is well-known that the quadratic form
 \( \boldsymbol \psi ^* \boldsymbol \rho \, \boldsymbol \psi \) 
 	is positive semi-definite if and only if all the diagonal elements of $ \mathbf D$ 
	are non-negative.
	Since the diagonal elements of $ \mathbf D$ are the eigenvalues of $ \boldsymbol \rho $,
	it follows that the matrix $ \boldsymbol \rho $ is positive semi-definite
	if and only if $ \lambda \ge 0$ for all $ \lambda \in s^{ \boldsymbol \rho }$.
	\item Applying Proposition \ref{prop:trace} to the matrix $ \boldsymbol \rho $ 
	rather than to $ \mathbf A$ gives 
 \( \tr \boldsymbol \rho = \sum _{ \lambda \in \Lambda } m_\lambda \lambda \) .
	It follows that $ \tr \boldsymbol \rho = 1$	if and only if \eqref{eq:traceOne} is satisfied.
\end{enumerate}
	The proposition follows immediately.
\end{proof}

\subsection{Quantum States as Ortho-Probability Measures} \label{ss:orthoProbs}

This section shows how any quantum state, 
in the form of a density matrix $ \boldsymbol \rho $ on $ \C^n $, 
can be identified with an ``ortho-probability'' measure 
over the components of an appropriate orthogonal decomposition of $ \C^n $
into the different eigenspaces of $ \boldsymbol \rho $.   
The following two definitions appear to be novel.

\begin{definition} \label{def:probLabelledOrthDecomp}
	Say that the orthogonal decomposition $ \bigoplus _{ \lambda \in \Lambda } L_\lambda $ 
	of $ \C^n $	is \emph{numerically identified by probability} 
	just in case the finite set $ \Lambda $ of numerical identifiers are all non-negative
	and, together with the dimensions $ m_\lambda $ of each subspace $ L_\lambda $, 
	satisfy the unit trace equation \eqref{eq:traceOne}.
\end{definition}

\begin{definition} \label{def:orthProb}
	Consider any probability space $( \C^n, \Sigma ^D, \mathbb P^D )$ 
	with $ \Sigma ^D $ as the ortho-algebra on $ \C^n $
 	generated by the orthogonal decomposition $ \bigoplus _{d \in D} L_d $,
	as specified in Definition \ref{def:orthoAlgebra},
	and with $ \mathbb P^D $ as a probability measure on $ \Sigma ^D $.

	The triple $( \C^n, \Sigma ^D, \mathbb P^D )$ is an \emph{ortho-probab\-ility space}
	with \emph{ortho-probab\-ility measure} $ \mathbb P^D $ just in case, 
	given the dimension $ m_d := \dim L_d \in \N$ for each $d \in D$, there exists a bijection 
 \( D \owns d \longleftrightarrow \lambda _d \in \Lambda \subset [0, 1] \) such that:
\begin{enumerate}
	\item $ \sum _{d \in D} m_d \, \lambda _d = 1$;
	\item for all $d \in D$, one has $ \mathbb P^D ( \hat L_d ) = m_d \, \lambda _d $;
	\item for the residual set $ R^D = \C^n \setminus \cup _{d \in D} \hat L_d $
	one has $ \mathbb P^D ( R^D ) = 0$.
\end{enumerate}
\end{definition}

Note that requiring the mapping $d \longleftrightarrow \lambda _d $ to be a bijection 
ensures that the orthogonal decomposition $ \bigoplus _{d \in D} L_d $
is numerically identified by its probability, with each set $ \hat L_d $ 
as the unique eigenspace corresponding to the eigenvalue $ \lambda _d $.     

Here is the second main result of the paper.

\begin{theorem} \label{thm:orthProb}
	There are natural bijections between the sets of:
	\begin{enumerate}
		\item $n \times n$ density matrices $ \boldsymbol \rho $;
		\item orthogonal decompositions $ \bigoplus _{ \lambda \in \Lambda } L_\lambda $ 
		of $ \C^n $ that are numerically identified by probability;
		\item ortho-probability spaces $( \C^n, \Sigma ^D, \mathbb P^D )$.	
	\end{enumerate}
\end{theorem}

\begin{proof}
	Given the spectrum $ \Lambda = s^{ \boldsymbol \rho }$ 
	of any Hermitian matrix $ \boldsymbol \rho $, 
	we rewrite the bijections in \eqref{eq:bijections} as
\begin{equation} 
	\boldsymbol \rho = \sum\nolimits _{ \lambda \in s^{ \boldsymbol \rho }} 
	\lambda \, \boldsymbol \Pi _{ L_\lambda } 
	\longleftrightarrow 
	\bigoplus\nolimits _{ \lambda \in \Lambda } L_\lambda 
	\longleftrightarrow 
	\Gamma \left( \bigoplus\nolimits _{ \lambda \in \Lambda } L_\lambda \right)
\end{equation}
	Here these are the bijections specified in Theorem \ref{thm:threeBijections} 
	between the three spaces of:
	(i) spectrally decomposed Hermitian matrices; 
	(ii) numerically identified orthogonal decompositions;
	and (iii) restricted graphs of ortho-measurable functions, 
	as defined in Theorem \ref{thm:threeBijections}.
	
	By Proposition \ref{prop:densityMatrix}, the Hermitian matrix $ \boldsymbol \rho $
	is a density matrix if and only if all its eigenvalues $ \lambda $:
	(i) are non-negative;
	(ii) and together satisfy the unit trace equation \eqref{eq:traceOne}.
	But this double condition exactly matches both: 
	(i) the same double condition on the spectrum $ \Lambda = s^{ \boldsymbol \rho }$ 
	as that used in Definition \ref{def:probLabelledOrthDecomp} 
	of an orthogonal decomposition \( \bigoplus\nolimits _{ \lambda \in \Lambda } L_\lambda \) 
	which is numerically identified by probability;
	(ii) the double condition used in Definition \ref{def:orthProb} 
	of an ortho-probability space.
\end{proof}

\subsection{Mixed versus Pure Quantum States} \label{ss:pureStates}

We have just shown how a general quantum state, 
as represented by a density matrix $ \boldsymbol \rho $, 
can be identified with an ortho-probability measure $ \mathbb P^{ \boldsymbol \rho }$ 
defined on the ortho-algebra $ \Sigma ^{ \boldsymbol \rho }$ generated 
by the orthogonal decomposition $ \bigoplus _{ \lambda \in s^{ \boldsymbol \rho }} E_\lambda $
of $ \C^n$ into the eigenspaces of $ \boldsymbol \rho $.
In general, one may regard $ \mathbb P^{ \boldsymbol \rho }$
as describing a \emph{mixed} quantum state.
The special case of a \emph{pure} quantum state occurs 
when the spectrum of $ \boldsymbol \rho $ satisfies $ s^{ \boldsymbol \rho } = \{0, 1\} $
and there exists a unique normalized wave vector $ \boldsymbol \phi \in \Sph $
such that the spectral decomposition 
\( \boldsymbol \rho = \sum _{ \lambda \in s^{ \boldsymbol \rho }} 
	\lambda \, \boldsymbol \Pi _{ L_\lambda } \) of $ \boldsymbol \rho $ 
reduces to the single non-zero term
\( \boldsymbol \rho = \boldsymbol \Pi _{[ \boldsymbol \phi ]} 
 = \boldsymbol \phi \, \boldsymbol \phi ^* \)
involving the projection $ \boldsymbol \Pi _{[ \boldsymbol \phi ]}$
onto the one-dimensional subspace $[ \boldsymbol \phi ]$ spanned by~$ \boldsymbol \phi $.

Consider a general measurable space $( \Omega, \mathcal A)$ in which, 
for each $ \omega \in \Omega $, one has $ \{ \omega \} \in \mathcal A$.  
Then there is evidently a convex set 
of possible probability measures $ \mathbb P$ on $( \Omega, \mathcal A)$ 
whose extreme points are the degenerate probability measures $ \delta _\omega $ 
on $( \Omega, \mathcal A)$ which, for each $ \omega \in \Omega $, 
satisfy $ \delta _\omega ( \{ \omega \} ) = 1$.

By contrast, consider the multi-measurable space $( \Omega, ( \Sigma ^D )_{D \in \mathcal D}) $ where each $ \sigma $-algebra $ \Sigma ^D $ may be a based on a different 
orthogonal decomposition $ \bigoplus _{d \in D} L_d $ of $ \C^n $.
Consider the mixture \( \boldsymbol \rho = \sum _{k = 1}^m \alpha _k \boldsymbol \rho _k \)
of the set $ \{ \boldsymbol \rho _k \} _{k = 1}^m $ of~$m$ different 
pure states $ \boldsymbol \rho _k = \boldsymbol \Pi _{[ \boldsymbol \phi _k ]} $, 
where $ \alpha _k > 0$ for $k = 1, 2, \ldots, m$ and $ \sum _{k = 1}^m \alpha _k = 1$. 
For~$ \boldsymbol \rho $ to be a mixed state in the sense of a density matrix, 
and so an ortho-probability measure,
there must exist one single common orthogonal decomposition $ \bigoplus _{d \in D} L_d $ 
of $ \C^n $ which includes all the spaces $[ \boldsymbol \phi _k ]$.
This implies mutual orthogonality of all the normalized wave vectors $ \boldsymbol \phi _k $, 
as well as of the corresponding projections $ \boldsymbol \Pi _{[ \boldsymbol \phi _k ]}$
onto the one-dimensional subspaces that they span.

\subsection{Consistent Multi-Probability Spaces}

The following definition applies to probability measures
the consistency notion that Vorob$'$ev (1962) applied to measures more generally.

\begin{definition} \label{def:vorob}
	Given the family $ \mathcal D$ of orthogonal decompositions of $ \C^n $,
	the multi-probability space \( ( \C^n, ( \Sigma ^D, \mathbb P^D )_{D \in \mathcal D}) \)
	is \emph{consistent}
	just in case, whenever $ D_1 $ and $ D_2 $ both belong to $ \mathcal D$,
	and the linear space $L$ has the property that $\hat L = L \setminus \{ \mathbf 0\} $ 
	belongs to both $ \sigma $-algebras $ \Sigma ^{ D_1 }$ and $ \Sigma ^{ D_2 }$,
	then the two contextual probability measures $\mathbb P^{ D_1 }$ and $ \mathbb P^{ D_2 }$
	satisfy \( \mathbb P^{ D_1 } ( \hat L) = \mathbb P^{ D_2 } ( \hat L) \).
\end{definition}

\subsection{Quantum Probability Distributions over Projections}

Let $ \mathcal L$ denote the set of all linear subspaces $L$ of $ \C^n $,
and $ \mathscr P := \{ \Pi _L \mid L \in \mathcal L \} $ the domain 
of all orthogonal projections onto subspaces $L$ of $ \C^n $.

\begin{definition} \label{def:quantProb}
A \emph{quantum probability distribution} is a mapping 
 \( \mathscr P \owns \Pi \mapsto \mu ( \Pi ) \in [0, 1] \) with the property that,
	if $ \{ \Pi _k \}_{k \in K} $ is a family of mutually orthogonal projection matrices, 
	then $ \mu \left( \sum _{k \in K} \Pi _k \right) = \sum _{k \in K} \mu ( \Pi _k ) $.%
\footnote{This is a finite-dimensional version
of the definition on p.~31 of Parthasarathy (1992).}
\end{definition}

\subsection{Characterizing Consistent Multi-Probability Spaces}

\begin{proposition}
	The multi-probability space \( ( \C^n, ( \Sigma ^D, \mathbb P^D )_{D \in \mathcal D}) \)
	is consistent if and only if there exists a quantum probability distribution
	 \( \mathscr P \owns \Pi \mapsto \mu ( \Pi ) \in [0, 1] \)
	such that, whenever $L \in \Sigma ^D $, then $ \mathbb P^D (L) = \mu ( \Pi _L )$,
	independent of the contextual orthogonal decomposition $D$.
\end{proposition}

\begin{proof}
	The result is an immediate implication 
	of the two Definitions \ref{def:vorob} and \ref{def:quantProb}.
\end{proof}

Suppose the dimension $n$ of $ \C^n $ satisfies $n \ge 3$.
Then a corollary of Gleason's (1957) theorem due to Parthasarathy (1992, Theorem 9.18)
implies that there is a density matrix $ \boldsymbol \rho $ 
satisfying $ \mu ( \Pi ) = \tr ( \boldsymbol \rho \, \Pi )$ 
for all projections $ \Pi \in \mathscr P$.
This implies the \emph{trace rule} stating that,
for all $L \in \mathcal L$ and all $D \in \mathcal D$, 
one has $ \mathbb P^D (L) = \mu ( \Pi _L ) = \tr ( \boldsymbol \rho \, \Pi _L ) $,
which the \emph{Hilbert--Schmidt inner product} 
of the two matrices $ \boldsymbol \rho $ and $ \Pi _L $.  

\section{Conclusion} \label{s:conclude}
\subsection{A Quantum Measurement Tree} \label{ss:qmt}

Shafer and Vovk (2001, pp.\ 189--191) present 
a convenient and concise mathematical description 
of a typical simple quantum experiment based on a finite-dimensional Hilbert space $ \C^n $
where observables can be represented by Hermitian matrices rather than self-adjoint operators.
Where time is not explicitly involved, 
this paper has laid out an alternative and possibly more informative description
in the form of a very simple quantum measurement tree.
This takes the form of a tree graph in which each path through the tree has three nodes, 
as follows.

\begin{enumerate}
	\item First there is an initial \emph{preparation node} $ n_0 $ 
	where an experimental configuration or context $c$ in a finite domain $C$ is determined.%
\footnote{We insist on the domain $C$ of contexts being finite only to ensure
that the measurement tree is finite, in the sense of having a finite set of nodes.
Extensions allowing $C$ to be infinite only present difficulties if,
for instance, we want to introduce a general measurable random process 
which determines the experimental configuration $c \in C$.}  
	Each context $c \in C$ combines: 
	\begin{itemize}
		\item a \emph{quantum observable} in the form of a Hermitian matrix $ \mathbf A$
		whose eigenvectors and eigenvalues determine 
		a numerically identified orthogonal decomposition of $ \C^n $,
		with an associated \emph{ortho-algebra} of ortho-measur\-able events,
		and whose spectral decomposition corresponds to an \emph{ortho-measur\-able function}
		\( \C^n \owns \boldsymbol \psi \mapsto f( \boldsymbol \psi ) \in \R \)
		that associates each eigenvalue $ \lambda \in s^\mathbf A$ 
		with its \emph{eigenspace} $ E_\lambda = f^{-1} ( \lambda )$ 
		of corresponding eigenvectors;
		\item a \emph{quantum density matrix} $ \boldsymbol \rho $ in the form 
		of an orthogonal decomposition of $ \C^n $ identified by probability. 
	\end{itemize}
	\item Second, the immediate successors of the initial node $ n_0 $ consist, 
	for each contextual pair $( \mathbf A, \boldsymbol \rho )$, 
	a unique \emph{measurement node} $ n_{ \mathbf A, \boldsymbol \rho }$
	in the form of a chance node where there is a roulette lottery, 
	with a probability mass function $ \mathbb P^{ \mathbf A, \boldsymbol \rho }$ 
	defined on the finite set $ \{ E_\lambda \mid \lambda \in s^\mathbf A \}$ 
	of eigenspaces of~$ \mathbf A$ by the \emph{trace rule} requiring that 
 \( \mathbb P^{ \mathbf A, \boldsymbol \rho } ( E_\lambda ) 
 = \tr ( \boldsymbol \rho \, \Pi _{E _\lambda } ) \)
	for each eigenvalue $ \lambda \in s^\mathbf A$.
	\item Third, the immediate successors 
	of each measurement node $ n_{ \mathbf A, \boldsymbol \rho }$ consist, 
	for each potential value $ \lambda \in s^\mathbf A$ of the observable $ \mathbf A$, 
	of a unique terminal \emph{observable node} $ n_{ \mathbf A, \boldsymbol \rho, \lambda }$ 
	of the tree that is identified with the potential observation $ \lambda \in s^\mathbf A$.
\end{enumerate}

\subsection{Reduction to a Probability Metaspace}

In Hammond (2025) two particular multi-probability spaces were each reduced 
to a single probability metaspace whose sample space included a variable $ \sigma $-algebra.
Each reduction involved a finite domain $C$ of possible contexts, 
together with a probability mass function $C \owns c \mapsto q_c \in [0, 1] $
in the set $ \Delta (C) $ of those mappings that satisfy $ \sum _{c \in C} q_c = 1$.
Applying a similar reduction to the quantum measurement tree described in Section \ref{ss:qmt}
results in a probability metaspace $( \Omega ^M, \Sigma ^M, \mathbb P^M_q )$ where:
\begin{enumerate}
	\item The sample space $ \Omega ^M $ 
	is the Cartesian product set $ \C^n \times C$ of pairs $( \boldsymbol \psi, c)$ 
	that combine a wave vector $ \boldsymbol \psi \in \C^n $
	with a context $c$ chosen from the finite set $C$ 
	that determines a Hermitian matrix $ \mathbf A_c $ 
	which induces an ortho-algebra $ \Sigma _c = \Sigma ^{ \mathbf A_c }$ on $ \C^n $,
	including a residual set $ R_c = R ^{ \mathbf A_c }$.
	\item Let $ \mathfrak P^M $ be the partition of $ \Omega ^M $ into the finite collection
\begin{equation} \label{eq:partitionM}
	\cup _{c \in C} \left[ \cup _{E \in \Sigma _c } (E \times \{c\} ) \right]
\end{equation}
 	of cells which, for each context $c \in C$ 
	and then for every ortho-measurable set $E \in \Sigma _c $, 
	take the form of the Cartesian product set $E \times \{c\} $. 
	Then $ \Sigma ^M $ is the $ \sigma $-algebra $ \sigma ( \mathfrak P^M )$ 
	generated by the cells of the partition $ \mathfrak P^M $, 
	following the construction specified in Section \ref{ss:orthoMeas}.
	\item For each probability mass function $q \in \Delta (C) $, 
	the probability measure $ \mathbb P^M_q $ is the unique extension 
	to the finite $ \sigma $-algebra $ \Sigma ^M $ of the probability mass function which, 
	for each context $c \in C$ 
	and associated orthogonal decomposition $ \bigoplus _{ \lambda \in \Lambda _c } L_\lambda $
	into eigenspaces of $ \mathbf A_c $, satisfies $ \mathbb P^M_q ( R_c \times \{c\} ) = 0$ 
	for the residual set $ \R_c = \C^n \setminus \cup _{ \lambda \in \Lambda _c } \hat L_\lambda $,
	as well as 
 \( \mathbb P^M_q ( L_\lambda \times \{c\} ) = q_c \tr ( \boldsymbol \rho \, \Pi _{ L_\lambda }) \)
	for each $ \lambda \in \Lambda _c $.
\end{enumerate}

\subsection{Concluding Remarks and Future Research}

Consider any quantum observable represented by the $n \times n$ Hermitian matrix $ \mathbf A$. 
Note that each eigenvalue $ \lambda $ in the spectrum $ s^\mathbf A$ of $ \mathbf A$
is only a \emph{potential} observation.
Even if the eigenspace $ E_\lambda $ is the realized result of the roulette lottery
at the measurement node $ n_{ \mathbf A, \boldsymbol \rho }$ 
of the quantum measurement tree described in Section \ref{ss:qmt},
whether the eigenvalue $ \lambda $ is \emph{actually} observed depends 
on whether the experimental configuration includes some device for making that observation.

Note too that in this paper the roulette lotteries 
with probabilities specified by the trace rule have been limited to those
with a finite set of outcomes in the real line,
whose corresponding measurement operators are Hermitian matrices.
In future work it is planned to consider both:
(i) more general real-valued observables 
with an infinite range of possible values, possibly unbounded;
(ii) vector-valued observables 
with a range of multidimensional simultaneously observed measurements, 
whose different components are described by commuting Hermitian matrices.

Finally, one reason to consider quantum measurement trees is that, in principle, 
they should be able to describe the effects of a sequence of quantum measurements.
Then, of course, it is important to model how any quantum measurement affects
the unitary evolution of a quantum state.
This is especially true when, on the basis of Theorem \ref{thm:orthProb},
we regard any quantum state as an ortho-probability measure 
on the components of the appropriate orthogonal decomposition of $ \C^n $
into the eigenspaces of the usual Hermitian density matrix.

\medskip
\begin{center}
\textbf{Acknowledgements}:
\end{center}\noindent
This is an extensively revised and corrected version
of a paper based on the second half of a talk in March 2023 entitled
``Quantum Observables, Contextual Boolean Algebras, 
and Bayesian Rationality in Decision Trees''.
This was presented to the Workshop 
on the Applications of Topology to Quantum Theory and Behavioral Economics,
held at the Fields Institute for Research in Mathematical Sciences in Toronto.
This latest revision owes much to two anonymous referees 
and to two editors Emmanuel Haven and Jerome Busemyer, 
whose very helpful remarks did much to encourage and facilitate 
some essential corrections as well as significant improvements.
Finally, I repeat the earlier acknowledgements in the companion paper Hammond (2025).

\newpage
\begin{center} \textbf{References} 
\end{center}

\begin{description}
\small

\item Anscombe, Frank J., and Robert J. Aumann (1963)	
	``A Definition of Subjective Probability''
	\textit{Annals of Mathematical Statistics} 34: 199--205.

\item Friedberg, Richard M., and Pierre C. Hohenberg (2014) ``Compatible Quantum Theory''
	\textit{Reports on Progress in Physics} 77: 092001.

\item Gleason, Andrew M. (1957)
	``Measures on the Closed Subspaces of a Hilbert Space''
	\textit{Journal of Mathematics and Mechanics} 6 (4): 885--893.

\item Griffiths, Robert B. (2002) \textit{Consistent Quantum Theory}
	(Cambridge University Press).
	
\item Hammond, Peter J. (1988) ``Consequentialist Foundations for Expected Utility''
	\textit{Theory and Decision} 25: 25--78.

\item Hammond, Peter J. (2022) 
	``Prerationality as Avoiding Predictably Regrettable Consequences''
	\textit{Revue \'Economique} 73: 943--976.
	
\item Hammond, Peter J. (2025) ``Quantum Measurement Trees, I: 
	Two Preliminary Examples of Induced Contextual Boolean Algebras'' 
	University of Warwick, Centre for Research in Economic Theory and Applications (CRETA),
	Working Paper No.\ 90; 
	forthcoming in \textit{Philosophical Transactions of the Royal Society A}.

\item Hohenberg, Pierre C. (2010) 
	``Colloquium: An Introduction to Consistent Quantum Theory''
	\textit{Review of Modern Physics} 82 (4): 2835--2844.

\item Parthasarathy, Kalyanapuram Rangachari (1992)
	\textit{An Introduction to Quantum Stochastic Calculus} (Basel: Birkh\"auser Verlag).

\item Raiffa, Howard (1968) 
	\textit{Decision Analysis: Introductory Lectures on Choices under Uncertainty}
	(Addison-Wesley).

\item Savage, L.J. (1954, 1972) \textit{Foundations of Statistics}
	(New York: John Wiley; and New York: Dover Publications).
	
\item Shafer, Glenn, and Alexander Vovk (2001) 
	\textit{Probability and Finance: It's Only a Game!} (Wiley).

\item Von Neumann, John (1928) ``Zur Theorie der Gesellschaftsspiele'' 
	\textit{Mathematische Annalen} 100: 295--320.
	
\item Vorob$'$ev, Nikolai N. (1962) 
	``Consistent Families of Measures and Their Extensions''
	\textit{Theory of Probability and its Applications} 7: 147--163.

\end{description}
\end{document}